\documentclass[letterpaper]{article}

\usepackage[utf8]{inputenc} 
\usepackage[T1]{fontenc}
\usepackage{verbatim}
\usepackage{fullpage}

\usepackage[colorlinks=true,linkcolor=blue,urlcolor=blue,citecolor=blue,anchorcolor=green,pdfusetitle]{hyperref}

\usepackage{amsmath}
\usepackage{amsfonts}
\usepackage{amssymb,amsthm,mathtools}

\usepackage{cleveref}

\usepackage{graphicx}
\usepackage{booktabs} 

\usepackage[backend=biber, sorting=nyt, style = numeric, doi = false, url = false, isbn = false, maxbibnames = 6,maxcitenames = 2]{biblatex}
\bibliography{references.bib}
\renewbibmacro{in:}{}      
\newbibmacro{string+doi}[1]{\iffieldundef{doi}{#1}{\href{https://dx.doi.org/\thefield{doi}}{#1}}}
\DeclareFieldFormat{title}{\usebibmacro{string+doi}{\mkbibemph{#1}}}
\DeclareFieldFormat[article]{title}{\usebibmacro{string+doi}{\mkbibquote{#1}}}
\DeclareFieldFormat[incollection]{title}{\usebibmacro{string+doi}{\mkbibquote{#1}}}                   
\DeclareFieldFormat[inproceedings]{title}{\usebibmacro{string+doi}{\mkbibquote{#1}}}  
\newtheorem{thm}{Theorem}
\newtheorem{prop}[thm]{Proposition}
\newtheorem{lem}[thm]{Lemma}
\newtheorem{cor}[thm]{Corollary}
\theoremstyle{definition}
\newtheorem{rem}[thm]{Remark}
\newtheorem{defn}[thm]{Definition}
\newtheorem{ex}[thm]{Example}

\usepackage{braket}
\usepackage{mathtools}
\usepackage{dsfont}
\usepackage{xcolor}

\newcommand{\cE}{\mathcal{E}}
\newcommand{\cF}{\mathcal{F}}
\newcommand{\cH}{\mathcal{H}}

\newcommand{\cU}{\mathcal{U}}
\newcommand{\cV}{\mathcal{V}}
\newcommand{\cW}{\mathcal{W}}

\newcommand{\mc}{\mathcal}
\newcommand{\mbb}{\mathbb}

\newcommand{\cB}{\mathcal{B}}

\DeclareMathOperator{\tr}{tr}

\DeclareMathOperator{\supp}{supp}

\newcommand{\lmax}{\lambda_{\mathrm{max}}}
\newcommand{\lmin}{\lambda_{\mathrm{min}}}

\newcommand{\one}{\mbb{I}}

\newcommand{\psucc}{p_{\mathrm{succ}}}

\newcommand{\ppgm}{p_{\mathrm{PGM}}}

\newcommand{\cT}{\mathcal{T}}

\usepackage{xcolor}

\definecolor{cool_green}{rgb}{0.0, 0.5, 0.0}

\usepackage[affil-it]{authblk}

\begin{document}
	
\title{On the distinguishability of\\geometrically uniform quantum states}
\author[1]{Juntai Zhou\thanks{Corresponding author. \href{mailto:juntaiz2@illinois.edu}{juntaiz2@illinois.edu}}}
\author[2]{Stefano Chessa\thanks{\href{mailto:schessa@illinois.edu}{schessa@illinois.edu}}}
\author[2,3]{Eric Chitambar\thanks{\href{mailto:echitamb@illinois.edu}{echitamb@illinois.edu}}}
\author[1,3]{Felix Leditzky\thanks{\href{mailto:leditzky@illinois.edu}{leditzky@illinois.edu}}}

\affil[1]{\small Department of Mathematics, University of Illinois Urbana-Champaign}
\affil[2]{Department of Electrical and Computer Engineering, University of Illinois Urbana-Champaign}
\affil[3]{Illinois Quantum Information Science and Technology (IQUIST) Center, University of Illinois Urbana-Champaign}

\maketitle

\begin{abstract}
	A geometrically uniform (GU) ensemble is a uniformly weighted quantum state ensemble generated from a fixed state by a unitary representation of a finite group $G$. 
    In this work we analyze the problem of discriminating GU ensembles from various angles. 
    Assuming that the representation of $G$ is irreducible, we first show that a particular optimal measurement can be understood as the limit of weighted `pretty good measurements' (PGM).
    This naturally provides examples of state discrimination for which the unweighted PGM is provably sub-optimal. 
    We extend this analysis to certain reducible representations, and use Schur-Weyl duality to discuss two particular examples of GU ensembles in terms of Werner-type and permutation-invariant generator states. 
    For the case of pure-state GU ensembles we give a streamlined proof of optimality of the PGM first proved in [Eldar et al., 2004]. 
    We use this result to give simplified proofs of the optimality of the PGM, along with expressions for the corresponding success probabilities, for two tasks: 
    the hidden subgroup problem over semidirect product groups (first proved in [Bacon et al., 2005]), and port-based teleportation (first proved in [Mozrzymas et al., 2019] and [Leditzky, 2022]). 
    Finally, we consider the $n$-copy setting and adapt a result of [Montanaro, 2007] to derive a compact and easily evaluated lower bound on the success probability of the PGM for this task. 
    This result can be applied to the hidden subgroup problem to obtain a new proof for an upper bound on the sample complexity by [Hayashi et al., 2006].
\end{abstract}

\section{Introduction}

The problem of quantum state discrimination was studied even before the dawn of modern quantum information theory \cite{Helstrom-1976a,Holevo-2011a}.  The general task is to identify one of many possible state preparations for a given quantum system by measuring it in a judiciously chosen way.  From the experimenter's perspective, the possible states are described by a quantum state ensemble, $\mc{E}=(p_i,\rho_i)_{i=1}^N$, such that the given system has been prepared in state $\rho_i$ with probability $p_i$.  There are a variety of ways to quantify how well an experimenter can achieve the goal of identifying the prepared state.
Historically, much attention was given to maximizing the accessible information of the ensemble, which is the maximal mutual information attainable between the unknown state variable $i$ and the random variable of the experimenter's measurement outcome \cite{Fuchs-1996a}.  
This quantity is related to classical capacities of a quantum channel \cite{shor2003capacities,wilde2013quantum}.

A somewhat more user-friendly figure of merit, and the quantifier studied in this work, is the minimum average error probability in correctly guessing the state variable $i$ given the measurement outcome. 
Many quantum information- and learning-theoretic tasks are linked to quantum state discrimination in this minimum average error setting, for example randomness extraction \cite{koenig2009operational}, characterizing the fidelity of teleportation protocols \cite{ishizaka2008asymptotic,ishizaka2009quantum,chitambar2024teleportation}, bounding the sample complexity in learning tasks \cite{arunachalam2018optimal,arunachalam2020quantumcoupon,hadiashar2024optimal}, or solving computational problems such as the hidden subgroup problem \cite{bacon2005semidirect,bacon2006dihedral,childs2010algebraic}.
Mathematically, computing the minimum average error is equivalent to solving a semidefinite program \cite{watrous2018theory,siddhu2022five}, a type of convex optimization problem with efficient numerical solvers and a useful duality theory that is heavily used in this paper.
For a comprehensive overview of the quantum state discrimination problem, we encourage the reader to consult the excellent reviews \cite{Bergou-2004a, Barnett-2009a, Bae-2015a}.

Quantum state discrimination becomes much more tractable for ensembles with a high degree of symmetry. 
A particular type of symmetry is geometric uniformity \cite{ban1997optimum,eldar2001quantum,eldar2004optimal}: An ensemble $\mc{E}$ is said to be geometrically uniform if its elements can be generated by the action of some finite group.  Specifically, if $G$ is a group with unitary representation $g$, then $\mc{E}$ is geometrically uniform with respect to $G$ if its elements have uniform prior probability, and they can each be expressed as $\rho_g \coloneqq g\rho g^\dagger$ for some fixed generator state $\rho$ and $g\in G$.
Geometrically uniform (GU) ensembles can be found in many quantum information-theoretic tasks such as port-based teleportation, certain instances of the hidden subgroup problem, or the quantum coupon collector problem (see Examples \ref{ex:pbt}, \ref{ex:dihedral-hsg}, and \ref{ex:coupon}, respectively).  
Davies first recognized that if a GU ensemble is generated by an irreducible unitary representation \cite{Davies-1978a}, then its accessible information can be attained using a positive operator-valued measure (POVM) $\{\Pi_g\}_g$ that can also be generated by the action of the group, i.e. $\Pi_g= g\Pi g^\dagger$ for some $\Pi\geq 0$.  This property was later shown by Sasaki \textit{et al.} to also hold for GU ensembles whose generating group acts invariantly on just the space of real vectors \cite{Sasaki-1999a}, and Decker generalized the optimal POVM structure to be multiple orbits of the group action when dealing with reducible representations \cite{Decker-2009a}.  For the problem of minimum error state discrimination, Eldar \textit{et al.} proved that an optimal strategy can likewise be obtained by a covariant POVM of the form $\Pi_g= g\Pi g^\dagger$ for arbitrary group actions \cite{eldar2004optimal}. 
For pure-state GU ensembles, \textcite{krovi2015optimal} derived a representation-theoretic formula for the success probability and applied it to state ensembles arising in optical communication problems.
Despite this progress, however, there are many important GU ensembles whose state distinguishability can be analyzed even further by focusing on their particular types of symmetry.  The goal of this paper is to sketch the mathematical framework for performing such an analysis. 

\subsection{Main results and structure of the paper}

In this paper we discuss quantum state discrimination of geometrically uniform (GU) ensembles.
In particular, we investigate the role of the ``pretty good measurement'' (PGM) or square-root measurement \cite{belavkin1975optimal,holevo1979asymptotically,hausladen1994pretty} and its generalization to ``power-PGMs'' \cite{tyson2009two-sided,tyson2009weighted} for discriminating GU ensembles, as well as general ensembles (in \Cref{sec:general-GU}).

We start in Section \ref{sec:preliminaries} with preliminaries related to quantum state discrimination.
We prove that the task of discriminating block-diagonal states can be reduced to first measuring the blocks and then discriminating the states within the blocks, possibly ignoring any redundant parts that may arise, such as in the Koashi-Imoto decomposition \cite{koashi2002,hayden2003,khanian2022mixed}.
This result is stated in \Cref{prop:block-diagonal} and \Cref{cor:block-diagonal}, and serves as a starting point for the following discussion.
We also provide some background on the hidden subgroup problem and port-based teleportation, tasks that are commonly studied through the lens of quantum state discrimination and discussed in later sections of this paper.

In Section \ref{sec:GU-ensembles} we formally introduce GU ensembles and present examples of this structure in the hidden subgroup problem, port-based teleportation, and the quantum coupon collector problem.

In Section \ref{sec:irreducible} we focus on GU ensembles defined via an irreducible representation of a finite group.
We prove that an optimal measurement is defined in terms of a subspace of the eigenspace corresponding to the largest eigenvalue of the generator, and we derive an explicit expression for the optimal success probability (\Cref{prop:geom uniform optimal}).
A special case of this optimal measurement can be understood as the $\alpha\to\infty$ limit of the $\alpha$-power-PGM (\Cref{prop:irrep-gu}), and we furthermore show that the success probabilities of the $\alpha$-power-PGMs are strictly increasing in $\alpha$.
This result gives examples where the (usual) PGM (with $\alpha=1$) is provably \emph{sub-optimal} (\Cref{sec:suboptimality-of-PGM}).
This analysis can also be applied to the state exclusion problem (\Cref{sec:exclusion}).

In Section \ref{sec:generators-symmetry} we generalize the above analysis of the optimal success probability and optimal measurement to certain GU ensembles defined in terms of a \emph{reducible} representation of a finite group. 
We exploit Schur-Weyl duality to give two examples of GU ensembles in terms of a Werner state generator (\Cref{prop:werner-generator}) and a permutation-invariant generator state (\Cref{prop:perm-inv-generator}).

In Section \ref{sec:pureGU} we turn to the special case of pure-state GU ensembles.
In this case, a result by \textcite{eldar2004optimal} shows that the PGM is the optimal measurement, and \textcite{dallapozza2015optimality} derived necessary and sufficient conditions for the optimality of the PGM for \emph{compound} pure-state GU ensembles consisting of the GU orbits of a collection of different generator states.
We give a simplified proof of the optimality of the PGM for pure-state GU ensembles in \Cref{prop:pure-GU-PGM-optimal} based on a simple expression for the PGM-success probability in this setting derived in \Cref{lem:GU-PGM}. 
We use this expression to give a new streamlined proof of the success probability of the coset state discrimination problem arising in hidden subgroup problems for semidirect products \cite{bacon2005semidirect}, and point out connections to the well-known Holevo-Curlander bounds on the optimal success probability, and the results of \textcite{hadiashar2024optimal} on the quantum coupon collector problem.
In \Cref{sec:pbt} we apply ideas from this discussion to discuss the optimal measurements for port-based teleportation.
This task involves a GU ensemble that does \emph{not} consist of pure states, yet the pure-state GU ensemble analysis developed above can still be applied when taking the natural symmetries of port-based teleportation into consideration.
We give a new proof of the optimality of the PGM for this task, which was previously shown in \cite{mozrzymas2018optimal,leditzky2020}.

In Section \ref{sec:general-GU} we return to the discussion of general GU ensembles.
We adapt a lower bound on the PGM-success probability by \textcite{montanaro2007distinguishability} to the GU setting to derive an analytical lower bound on the optimal success probability.
These arguments are applied to the hidden subgroup problem to give a new simplified proof of an upper bound on the sample complexity in solving this problem derived by \textcite{hayashi2006hiddensubgroup}, avoiding the use of the Hayashi-Nagaoka operator inequality \cite{hayashi2003general} altogether.
Finally, we make some concluding remarks in Section \ref{sec:conclusion}.

\section{Preliminaries}
\label{sec:preliminaries}

\subsection{Notation}
We denote by $\one$ the identity operator on a Hilbert space.
The smallest and largest eigenvalues of a Hermitian operator $X$ are denoted by $\lmin(X)$ and $\lmax(X)$, respectively.
The support of a linear operator acting on a Hilbert space is defined to be the orthogonal complement of its kernel.
For Hermitian operators $A,B$ we write $A\geq B$ if $A-B\geq 0$.
A quantum state is a linear positive semidefinite operator on a Hilbert space with unit trace.

\subsection{Quantum State Discrimination}
Let $\cE=(p_i,\rho_i)_{i=1}^N$ be a quantum state ensemble consisting of quantum states $\rho_1,\dots,\rho_N$ and a probability distribution $(p_1,\dots,p_N)$.
In the minimum average error setting, the optimal success probability $\psucc^*$ of distinguishing the states in $\cE$ is given by the following semidefinite program (SDP):
\begin{align}
	\psucc^* = \max\left\lbrace \sum\nolimits_{i=1}^N p_i \tr(M_i \rho_i) : M_i\geq 0 \text{ for all $i=1,\dots,N$}, \sum\nolimits_{i=1}^N M_i = \one\right\rbrace.
	\label{eq:primal}
\end{align}
The dual program gives the same value because of strong duality, and can be expressed as follows:
\begin{align}
	\psucc^* = \min\left\lbrace \tr K : K\geq p_i\rho_i \text{ for all $i=1,\dots,N$}\right\rbrace.
	\label{eq:dual}
\end{align}

A commonly considered measurement in state discrimination is the \emph{pretty good measurement }(PGM) or \emph{square-root measurement} \cite{belavkin1975optimal,holevo1979asymptotically,hausladen1994pretty}, defined as 
\begin{align}
	M_i=p_i\overline{\rho}^{-\frac{1}{2}}\rho_i\overline{\rho}^{-\frac{1}{2}} \label{eq:pgm}
\end{align} 
where $\overline{\rho}=\sum_i p_i\rho_i$ is the average state and the inverse is taken on the support of $\overline{\rho}$.
Without loss of generality, one can restrict the Hilbert space on which the states $\rho_i$ act to this support without changing the optimal value in \eqref{eq:primal} or \eqref{eq:dual}. 
One generalization of PGM is the $\alpha$-power-PGM \cite{tyson2009weighted,tyson2009two-sided}, defined for $\alpha\geq 1$ as
\begin{align}
	M_i^{(\alpha)}=\tilde{\rho}^{-\frac{1}{2}}(p_i\rho_i)^\alpha\tilde{\rho}^{-\frac{1}{2}} \label{eq:power-PGM}
\end{align}
where $\tilde{\rho}=\sum_i(p_i\rho_i)^\alpha$.
This measurement coincides with the PGM for $\alpha=1$, and with the \emph{quadratically-weighted measurement} from \cite{tyson2009two-sided} for $\alpha=2$.

Throughout the discussion, we will be using the following observation:
\begin{prop}
	\label{prop:block-diagonal}
    Let $\cH=\bigoplus_{j=1}^k V_j\otimes W_j$ be a direct-sum decomposition of a Hilbert space $\cH$ and let $(p_i,\rho_i)_{i=1}^{N}$ be an ensemble of quantum states on $\cH$ of the form
    \begin{align}
    \rho_i=\bigoplus_{j=1}^k q_{ij}\rho_{ij}\otimes\omega_j
    \end{align}
    for all $i=1,\dots,N$, where $\rho_{ij}$ and $\omega_j$ are quantum states for $j=1,\dots,k$. Then the discrimination task between $(p_i,\rho_i)_{i=1}^{N}$ can be reduced to discriminating between $(p_{ij},\rho_{ij})_{i=1}^{N}$ for each $j=1,\dots,k$, where $p_{ij}=\frac{p_iq_{ij}}{s_j}$ and $s_j=\sum_i p_iq_{ij}$. More precisely, denoting the optimal success probability of discriminating $(p_{ij},\rho_{ij})_{i=1}^{N}$ for fixed $j$ by $\psucc^*(j)$, we have 
    \begin{align}
    	\psucc^*=\sum_{j=1}^k s_j\psucc^*(j)
    \end{align}
    with optimal POVM
    \begin{align}
    	M_i=\bigoplus_{j=1}^k M_{ij}\otimes\one_{W_j}
    \end{align}
    where $(M_{ij})_{i=1}^{N}$ is an optimal POVM for discriminating $(p_{ij},\rho_{ij})_{i=1}^{N}$ for fixed $j=1,\dots,k$.
\end{prop}
\begin{proof}
    For each $j=1,\dots,k$ let $K_j$ be a feasible operator in the dual SDP problem for the ensemble $(p_{ij},\rho_{ij})_{i=1}^{n}$, i.e., $K_j\geq p_{ij}\rho_{ij}$ for all $i$ and $\psucc^*(j)=\tr(K_j)$. The operator
    \begin{align}
    	K\coloneqq \bigoplus_{j=1}^k s_jK_j\otimes\omega_j
    \end{align}
    satisfies
    \begin{align}
    	K\geq p_i\bigoplus_{j=1}^k q_{ij}\rho_{ij}\otimes\omega_j=p_i\rho_i\quad\forall\:i
    \end{align}
     and
     \begin{align}
     \tr(K)=\sum_{j=1}^k s_j\tr(K_j)=\sum_{j=1}^k s_j\psucc^*(j).
     \end{align}
     Therefore,
     \begin{align}
     \psucc^*\leq\tr(K)=\sum_{j=1}^k s_j\psucc^*(j).\label{eq:psucc-UB}
     \end{align}
     On the other hand, given an optimal POVM $(M_{ij})_{i=1}^{n}$ for each $j$ to discriminate $(p_{ij},\rho_{ij})_{i=1}^{n}$, let
     \begin{align}
     	M_i=\bigoplus_{j=1}^k M_{ij}\otimes\one_{W_j}.
     \end{align}
 	Note that $\sum_i M_{ij} = \one_{V_j}$, ensuring that $\sum_i M_i=\one$.
     Then,
     \begin{align}
     	\sum_i p_i\tr(M_i\rho_i)=\sum_i\sum_j p_iq_{ij}\tr(M_{ij}\rho_{ij})=\sum_j s_j\bigg(\sum_i p_{ij}\tr(M_{ij}\rho_{ij})\bigg)=\sum_j s_j\psucc^*(j),
     \end{align}
 	and hence $\psucc^* \geq \sum_j s_j\psucc^*(j)$.
 	This shows that $M_i$ is indeed an optimal choice of POVM achieving the upper bound \eqref{eq:psucc-UB}, and thus $\psucc^* = \sum_j s_j\psucc^*(j)$.
\end{proof}

\begin{cor}
\label{cor:block-diagonal}
    Discriminating a block-diagonal ensemble is equivalent to discriminating within each block.
    More precisely, suppose that $(p_i,\rho_i)_i$ is an ensemble of states with each $\rho_i$ having the same block-diagonal structure. 
    Then an optimal discrimination measurement can be built by first measuring with respect to the block-diagonal structure and projecting onto one of these blocks, and then discriminating optimally among the projected states within the blocks.
\end{cor}

\begin{proof}
    Simply let $\dim W_j=1$ in \Cref{prop:block-diagonal}, then the blocks consist of the $\rho_{ij}$, and the corollary follows.
\end{proof}

\begin{rem}
    The Koashi-Imoto (KI) decomposition \cite{koashi2002} has many applications such as the classification of the structure of quantum Markov chains \cite{hayden2003,sutter2018}, quantum source compression \cite{khanian2022mixed}, or the task of quantum state merging \cite{yamasaki2019}. 
    Given an ensemble $(p_i,\rho_i)_i$ on a Hilbert space $\cH$, there exists a decomposition of $\cH$ such that
    \begin{align}
        \rho_i\cong\bigoplus\nolimits_j q_{ij}\rho_{ij}\otimes\omega_j.\label{eq:koashi-imoto}
    \end{align}
    Here, $q_{ij}\geq0$, and the $\rho_{ij}$ and $\omega_j$ are density matrices. The decomposition classifies the degrees of freedom in the ensemble into three parts: the classical part $q_{ij}$, the non-classical part $\rho_{ij}$, and the redundant part $\omega_j$~\cite{khanian2022mixed}.
    That is, for any unitary $U$ satisfying $T_U(\rho_i)=\rho_i$ for all $i$ where $T_U(\rho)\coloneqq \tr_E[U(\rho\otimes\sigma_E)U^\dagger]$ for some state $\sigma_E$ on an environment $E$, we have
    \begin{align}
        U\cong\bigoplus\nolimits_j\one_j\otimes U_j
    \end{align}
    for suitable unitaries $U_j$.
    Since the redundant part $\omega_j$ in \eqref{eq:koashi-imoto} is independent of the index $i$, \Cref{prop:block-diagonal} implies that discriminating the ensemble $(p_i,\rho_i)$ can be reduced to discriminating the ensembles $(p_{ij},\rho_{ij})$ for each $j$.
    This gives another operational explanation of the ``redundant'' part.
\end{rem}

\subsection{Hidden subgroup problem}
The Hidden Subgroup problem (HSP) is an important class of computational problems generalizing factoring and graph isomorphism \cite{childs2010algebraic}.
Given a subgroup $H$ of a group $G$ and a function $f\colon G\to S$ such that $f(g_1)=f(g_2)$ if and only if $g_1\in g_2H$, the hidden subgroup problem asks for the generating set of $H$ by querying the function $f$. 
In the ``standard method'' \cite{childs2010algebraic} one starts with considering a unitary $U_f\colon |G|^{-1/2}\sum_{g\in G}\ket{g,0}\mapsto |G|^{-1/2} \sum_{g\in G}\ket{g,f(g)}$.
Discarding the second register yields the ``hidden subgroup state''
\begin{align}
	\rho_H=\frac{|H|}{|G|}\sum_{g\in K}\ket{gH}\bra{gH} \label{eq:hidden-subgroup-state}
\end{align}
where $K\subset G$ is a complete set of coset representatives of $H$ and $\ket{gH}\coloneqq |H|^{-1/2}\sum_{h\in H}\ket{gh}$. The question now reduces to asking for the sample complexity of discriminating the uniform ensemble of hidden subgroup states $\rho_H$ given a set of candidate subgroups $H$.

\subsection{Port-based teleportation}
\label{sec:prelim-pbt}
Port-based teleportation \cite{ishizaka2008asymptotic,ishizaka2009quantum} is a variant of standard quantum teleportation \cite{bennett1993teleporting} that operates on a large multipartite quantum state and crucially employs the simplified decoding operation of discarding all but one systems (or \emph{ports}) to achieve teleportation.
This particular decoding step equips port-based teleportation with full unitary covariance, causing the protocol to be necessarily imperfect using finite resources \cite{Nielsen_1997}.
However, the error of the protocol can be made arbitrarily small by increasing the entanglement in the resource state, and the unitary covariance property enables many interesting applications in cryptography \cite{Beigi_2011,dolev2022nonlocalcomputationquantumcircuits,dolev2022holographyresourcenonlocalquantum,May2022complexity}, non-locality \cite{Buhrman_2016}, channel discrimination and simulation \cite{Pirandola_2019,Pereira_2021}, asynchronous quantum information-processing \cite{kim2025resource}, and quantum programming \cite{ishizaka2008asymptotic,Banchi2020convex,Sedl_k_2019,Quintino2019probabilistic,Quintino2022deterministic,yoshida2024,grosshans2024multicopyquantumstateteleportation}.

The quality of port-based teleportation is conveniently measured in terms of the entanglement fidelity of the implemented teleportation channel.
In a general teleportation protocol, its entanglement fidelity is proportional to the success probability in a certain state discrimination problem that involves the uniformly weighted state ensemble obtained from applying the different decoding operations to the entangled resource state \cite{ishizaka2008asymptotic,ishizaka2009quantum,Beigi_2011,chitambar2024teleportation}.
We will analyze port-based teleportation through the lens of state discrimination in \Cref{ex:pbt} in the next section and \Cref{sec:pbt}.

\section{Geometrically uniform ensembles} \label{sec:GU-ensembles}
Let $\rho$ be a quantum state on a Hilbert space $\cH$ and $G$ be a finite group with a unitary representation on $\cH$.
To keep notation simple, we will just write $g$ for the representation matrix of the group element $g\in G$.
A \emph{geometrically uniform} (GU) ensemble $(\frac{1}{n},\rho_g)_{g\in G}$ (where $n=|G|$) consists of states $\rho_g = g\rho g^\dagger$ with uniform prior distribution $p_g=\frac{1}{n}$. 
More generally, let $\{\rho_1,...,\rho_m\}$ be a set of generator states and $G$ a finite group. A \emph{compound GU} (CGU) ensemble is given by $\rho_{g,k}=g\rho_k g^\dagger$ for $k=1,\dots,m$ and $g\in G$, with uniform prior distribution.

The symmetry of GU ensemble simplifies the state discrimination problem in many ways. For instance, there exists an optimal measurement generated by the group, i.e., $\Pi_g=g\Pi g^\dagger$ with generator $\Pi\geq 0$, satisfying $\sum_g \Pi_g = \one$.
In this case the optimal success probability is given by $p^*=\tr(\Pi\rho)$ \cite[Thm.~2]{eldar2004optimal}.  

For the pretty good measurement (PGM) in \eqref{eq:pgm}, using the $G$-invariance of the average state $\overline{\rho} = \frac{1}{n}\sum_{g\in G}\rho_g$,
    \begin{equation}\label{eq: ppgm}
        p_{\mathrm{PGM}}=\sum_{g\in G}\frac{1}{n}\tr\left (\overline{\rho}^{-\frac{1}{2}}\frac{1}{n}\rho_g\overline{\rho}^{-\frac{1}{2}}\rho_g \right)=\frac{1}{n^2}\sum_{g\in G}\tr\bigg((g^\dagger\overline{\rho}^{-\frac{1}{2}}g)\rho (g^\dagger\overline{\rho}^{-\frac{1}{2}}g)\rho\bigg)=\frac{1}{n}\tr\left(\overline{\rho}^{-\frac{1}{2}}\rho\overline{\rho}^{-\frac{1}{2}}\rho\right).
    \end{equation}
In particular, the PGM is always optimal for pure state GU ensembles \cite{eldar2004optimal}. 
We will give a simpler proof of this result by directly computing the success probability of the PGM in Section \ref{sec:pureGU}.

Many instances of state discrimination involve geometrically uniform ensembles, as the following selection of examples demonstrates.
    
\begin{ex}[Port-based teleportation]
\label{ex:pbt}
	An important variant of port-based teleportation (as introduced in \Cref{sec:prelim-pbt}) uses $n$ maximally entangled states $|\phi^+\rangle_{AB}$ as the resource state.
	More precisely, let $A\cong B\cong A_i\cong B_i\cong\mathbb{C}^d$ and consider $\phi_{A^nB^n}=(\phi^+_{AB})^{\otimes n}$, where $\phi^+_{AB} = \frac{1}{d}\sum_{i,j=1}^d \ket{ii}\bra{jj}$ is a maximally entangled state on $\mathbb{C}^d\otimes \mathbb{C}^d$.
    The marginal states $\rho_i\coloneqq\tr_{B_i^c}\phi_{A^nB^n}=\phi^+_{A_iB}\otimes\frac{1}{d^{n-1}}\one_{A_i^c}$ act on systems $A^nB_i$ (here, $A_i^c \equiv A_1\dots A_{i-1}A_{i+1}\dots A_n$, and similarly for $B_i^c$) and form a GU ensemble generated by the symmetric group $S_n$ acting on the systems $A^n$.
    As mentioned in \Cref{sec:prelim-pbt}, the success probability of discriminating this ensemble is proportional to the entanglement fidelity of the port-based teleportation channel with resource state $\phi_{A^nB^n}$.
    Ref.~\cite{leditzky2020} showed that the PGM is optimal in this discrimination problem, and the success probability is evaluated exactly, giving a new and explicit proof of an implicitly proved result in \cite{studzinski2017port}. We will give yet another proof of this result in \Cref{sec:pbt}, avoiding the heavy calculation of trace quantities using representation theory and the GU structure.
\end{ex}
\begin{ex}[Dihedral hidden subgroup problem]
\label{ex:dihedral-hsg}
    In the dihedral hidden subgroup problem \cite{bacon2006dihedral}, the hidden subgroup states \eqref{eq:hidden-subgroup-state} after quantum Fourier transformation are given for $d=0,\dots,N-1$ by
    \begin{align}
    	\rho_d=\frac{1}{N}\bigoplus_{k\in\mathbb{Z}_n}\ket{\tilde{\phi}_{k,d}}\bra{\tilde{\phi}_{k,d}}
    \end{align}
    where $\ket{\tilde{\phi}_{k,d}}=\frac{1}{\sqrt{2}}(\ket{0}+\omega^{kd}\ket{1})\ket{k}$, so each $\rho_d$ is block-diagonal where each block is a GU ensemble of pure states.
    It then follows from \Cref{prop:block-diagonal} that the PGM is an optimal measurement for this ensemble. 
    This can be generalized to any group of the form $A\rtimes\mathbb{Z}_p$ where $A$ is abelian and $p$ is prime \cite{bacon2005semidirect}. We will give a simpler computation of this case in Section \ref{sec:pureGU}.
\end{ex}
\begin{ex}
[Quantum coupon collector problem]
\label{ex:coupon}
    In the \emph{coupon collector problem} with given integers $n$ and $1< k<n$, the goal is to learn an unknown subset $S\subset [n]\coloneqq\lbrace 1,\dots,n\rbrace$ of size $k$.
    In the quantum version studied in \cite{arunachalam2020quantumcoupon,hadiashar2024optimal} one considers $n$-dimensional quantum systems, and the learner is given access to $t$ copies of a superposition over the $k$ elements in the unknown set $S$.
    The goal is again to learn $S$ with high probability, which corresponds to solving the discrimination problem associated with the uniformly weighted quantum state ensemble $(p_S,\phi_S)_{S\subset [n]}$, where $p_S = \binom{n}{k}^{-1}$ for all $S\subset [n]$ and $\phi_S = |\psi_S\rangle\langle\psi_S|^{\otimes t}$ with $|\psi_S\rangle = k^{-1/2}\sum_{s\in S}|s\rangle$.
    Consider the representation of $S_n$ on $\mathbb{C}^n$ via permuting the computational basis, i.e., representing $\pi\in S_n$ with the associated permutation matrix $\varphi_\pi$.
    Fixing some subset $S_0\subset [n]$, we have $\lbrace |\phi_S\rangle \rbrace= \lbrace (\varphi_\pi|\psi_{S_0}\rangle)^{\otimes t} : \pi\in S_n\rbrace$, and thus $(p_S,\phi_S)_{S\subset [n]}$ is a GU ensemble of size $n!/k!$ generated by $S_n$ (note that $|\psi_{S_0}\rangle$ is invariant under permuting basis vectors corresponding to elements in $S_0$).
\end{ex}
By twirling we can simplify the dual SDP \eqref{eq:dual} for discriminating GU ensembles. We will be using the following lemma throughout the paper:
\begin{lem}\label{lemma: GU SDP}
    We have
    \begin{align}
    	\psucc^*=\min_{\tilde{K}}\tr(\tilde{K}),
    \end{align}
    where the minimization is taken over all $G$-invariant operators $\tilde{K}$ satisfying $\tilde{K}\geq\frac{1}{n}\rho$.
\end{lem}
\begin{proof}
    First, if $\tilde{K}$ is $G$-invariant and $\tilde{K}\geq\frac{1}{n}\rho$, then $\tilde{K}\geq\frac{1}{n}\rho_g$ for all $g\in G$.
    Thus, $\tilde{K}$ is feasible in \eqref{eq:dual}, and hence
    \begin{align}
    	\psucc^*\leq\min_{\tilde{K}}\tr(\tilde{K}).
    \end{align} 
	On the other hand, let $K$ be an optimal operator that achieves the minimum, and let $\tilde{K}\coloneqq\frac{1}{n}\sum_{g\in G}gKg^\dagger$.
	Then, for $g\in G$,
    \begin{align}
    	\tilde{K}=\frac{1}{n}\sum_{g\in G}gKg^\dagger\geq\frac{1}{n}\sum_{g\in G}g(\frac{1}{n}\rho_{g^{-1}})g^\dagger=\frac{1}{n}\rho
    \end{align}
    and
    \begin{align}
    	\tr(\tilde{K})=\frac{1}{n}\sum_{g\in G}\tr(gKg^\dagger)=\tr(K).
    \end{align}
    Therefore,
    \begin{align}
    	\psucc^*=\tr(K)=\tr(\tilde{K})\geq\min_{\tilde{K}}\tr(\tilde{K}),
    \end{align}
	which concludes the proof.
\end{proof}

\section{Irreducible representations}
\label{sec:irreducible}

In this section, we assume that the generating group $G$ of the GU ensemble is a finite group of order $n=|G|$ that acts irreducibly on the Hilbert space $\cH$.
We will often make use of the identity
\begin{align}
	\sum_{g\in G} g X g^\dagger = \frac{n}{d}\tr(X) \one, \label{eq:schurs-lemma}
\end{align}
which follows from the irreducibility of the action of $G$ on $\cH$ and Schur's Lemma.

\subsection{Success probability and optimal measurement in GU state discrimination} \label{sec:optimal-success-probability}
An explicit formula for the success probability of discriminating a GU ensemble generated by an irreducible representation of a finite group was originally derived by Holevo \cite{holevo1973statistical}.
Here, we give a proof of this result in our framework.
\begin{prop}[\cite{holevo1973statistical}]\label{prop:geom uniform optimal}
    Let $G$ be a finite group ($|G|=n)$ with an irreducible representation $(U_g)_{g\in G}$ on a $d$-dimensional Hilbert space $\cH$, which we denote for simplicity by $g\equiv U_g$.
    Consider the GU ensemble $\lbrace \rho_g\rbrace_{g\in G}$ with $\rho_g = g\rho g^\dagger$ for a generator state $\rho$.
    Then the optimal success probability is equal to
    \begin{align}
        \psucc^* = \frac{d}{n} \lmax(\rho). \label{eq:irrep-optimal-psucc}
    \end{align}
    An optimal measurement can be constructed as follows: Let $d^{\max}$ denote the dimension of the eigenspace of $\rho$ corresponding to its largest eigenvalue, and choose $1\leq k\leq d^{\max}$ orthonormal vectors $\lbrace |\phi_1\rangle,\dots,|\phi_k\rangle\rbrace$ in this eigenspace.
    Setting
    \begin{align}
        \Pi = \frac{d}{nk} \sum_{i=1}^k |\phi_i\rangle\langle \phi_i|, \label{eq:irrep-optimal-measurement}
    \end{align}
    the measurement $\lbrace \Pi_g\rbrace_{g\in G}$ with $\Pi_g = g\Pi g^\dagger$ is optimal and achieves \eqref{eq:irrep-optimal-psucc}.
\end{prop}

\begin{proof}
    Definition \eqref{eq:irrep-optimal-measurement} of $\Pi$ and \eqref{eq:schurs-lemma} immediately show that $\lbrace \Pi_g\rbrace_{g\in G}$ is a valid POVM with success probability given by \eqref{eq:irrep-optimal-psucc}. 
    To show optimality, we define an operator 
    \begin{align}
        K\coloneqq \frac{1}{n} \sum_{g\in G} g \rho^{1/2} \Pi \rho^{1/2} g^\dagger = \frac{1}{d} \tr\left(\rho^{1/2}\Pi\rho^{1/2}\right) \one = \frac{1}{kn} \sum_{i=1}^k \tr(\rho |\phi_i\rangle\langle\phi_i|)\one = \frac{\lmax(\rho)}{n} \one,
    \end{align}
    where we once again used \eqref{eq:schurs-lemma} in the second equality.
    Since all $\rho_g$ are unitarily equivalent and thus have the same spectrum, we obtain that $K\geq \frac{1}{n}\rho_g$ for all $g\in G$, and thus $K$ is dual feasible in \eqref{eq:dual}, and its trace is equal to $\psucc^*$ in \eqref{eq:irrep-optimal-psucc}.
\end{proof}

\begin{prop}
\label{prop:irrep-gu}
    With the assumptions of \Cref{prop:geom uniform optimal}, the power-PGM defined in \eqref{eq:power-PGM} gives an optimal measurement in the limit $\alpha\to\infty$.
\end{prop}
\begin{proof}
    To see this, let first $\lambda\equiv \lmax(\rho)$, and denote by $\Pi_\rho^{\max}$ the projector onto the eigenspace of $\rho$ corresponding to $\lambda$.
    Observe that $\tilde{\rho} \coloneqq \frac{1}{\lambda}\rho$ satisfies $\lim_{\alpha\to \infty} \tilde{\rho}^\alpha = \Pi_\rho^{\max}$, since the rescaled eigenvalues in the other eigenspaces are strictly less than $1$ and therefore decay in the limit $\alpha\to\infty$. 
	Moreover, $\lambda = \lmax(\rho_g)$ for all $g\in G$.
	Setting $\Pi_g^{\max} = g \Pi_\rho^{\max}g^\dagger$, we then have
\begin{align}
    \lim_{\alpha \to \infty} M_g^{(\alpha)} & = 
    \lim_{\alpha \to \infty} \left( \sum\nolimits_{h\in G} \left(\frac{\rho_h}{n}\right)^\alpha \right)^{-\frac{1}{2}} \left(\frac{\rho_g}{n}\right)^\alpha \left( \sum\nolimits_{h\in G} \left(\frac{\rho_h}{n}\right)^\alpha \right)^{-\frac{1}{2}} \\
    &= \lim_{\alpha \to \infty} \left( \sum\nolimits_{h\in G} \rho_h^\alpha \right)^{-\frac{1}{2}} \rho_g^\alpha \left( \sum\nolimits_{h\in G} \rho_h^\alpha \right)^{-\frac{1}{2}} \\
    &= \lim_{\alpha \to \infty} \left( \sum\nolimits_{h\in G} h \tilde{\rho}^\alpha h^\dagger \right)^{-\frac{1}{2}} \tilde{\rho}_g^\alpha \left( \sum\nolimits_{h\in G} h \tilde{\rho}^\alpha h^\dagger \right)^{-\frac{1}{2}}\\
    & = \left( \sum\nolimits_{h\in G} h \Pi_\rho^{\max} h^\dagger \right)^{-\frac{1}{2}}  \Pi_g^{\max} \left( \sum\nolimits_{h\in G} h \Pi_\rho^{\max} h^\dagger \right)^{-\frac{1}{2}} \\
    &= 
    \left( \frac{nd^{\max}}{d} \right)^{-\frac{1}{2}} \Pi_g^{\max} \left( \frac{nd^{\max}}{d} \right)^{-\frac{1}{2}}\\
    &= \frac{d}{nd^{\max}}\Pi_g^{\max},
\end{align}
which concludes the proof.
\end{proof}

\subsection{Suboptimality of $\alpha$-PPGM for generic $\alpha$}\label{sec:suboptimality-of-PGM}
For a generic $\alpha$, the probability of success is equal to
\begin{align}
	\psucc =  \frac{1}{n}\sum_g \tr \left[ \rho_g M_g^{(\alpha)} \right]\; .
\end{align}

Assuming the irreducibility of the representation $G$ and using \eqref{eq:schurs-lemma}, we have that 
\begin{equation}
	M_g^{(\alpha)} = \left(\frac{\tr [\rho^{\alpha}]}{n^{\alpha-1} d}\one \right)^{-\frac{1}{2}} \left(\frac{\rho}{n}\right)^\alpha \left(\frac{\tr [\rho^{\alpha}]}{n^{\alpha-1} d}\one \right)^{-\frac{1}{2}} \; .
\end{equation}
Substituting above gives
\begin{equation}
	\psucc(\alpha) = \frac{n^{\alpha-1}}{n^{\alpha+1}} \frac{d}{\tr [\rho^{\alpha}] } \sum_g \tr \left[ \rho_g \rho_g^{\alpha} \right]
	= \frac{n^{\alpha-1}}{n^{\alpha}} \frac{d}{\tr [\rho^{\alpha}] }  \tr \left[ \frac{1}{n} \sum_g g \rho_g^{\alpha+1} g^\dagger \right]\; .
\end{equation}
Exploiting again irreducibility via \eqref{eq:schurs-lemma}, it follows that
\begin{equation}
	\psucc(\alpha) = \frac{d}{n}\frac{\tr [\rho^{\alpha+1}]}{\tr [\rho^{\alpha}]} \; .\label{eq:psucc-alpha}
\end{equation}
If the generator is pure then $\tr [\rho^{\alpha+1}] = \tr [\rho^{\alpha}]$ and we recover the optimality of the PGM already showed by Eldar \cite{eldar2004optimal}. If the generator is mixed, then the optimal probability reaches the maximum $ \psucc^* = d \lmax(\rho)/n$ (as proved in \Cref{prop:geom uniform optimal}) only in the limit $\alpha \to \infty$, that is, $\lim_{\alpha\to \infty}\psucc(\alpha) = \psucc^*$, which is easy to see from \eqref{eq:psucc-alpha}, and $\psucc(\alpha)<\psucc^*$ for all finite $\alpha$. This shows that for mixed generators the $\alpha$-PPGMs with finite $\alpha < \infty$ are strictly suboptimal.

\subsection{State exclusion problem}
\label{sec:exclusion}
We can generalize \Cref{prop:geom uniform optimal} to the \emph{state exclusion problem} \cite{pusey2012reality,bandyopadhyay2014exclusion,mcirvin2024prettybad}. In this problem, one aims to rule out the preparation of a specific state in an ensemble, which amounts to a minimization in the corresponding optimization problem for the error probability:
\begin{align}
	q^*=\min_{\Pi}\sum_i p_i\tr(\Pi_i\rho_i).
\end{align}
The dual problem is
\begin{align}
	q_*=\max \lbrace \tr(K) : K\leq p_i\rho_i \text{ for all }i\rbrace.
\end{align}
Using the twirling technique from above, for a GU ensemble we again have
\begin{align}
	q_*=\max_{\tilde{K}}\tr(\tilde{K})
\end{align}
where $\tilde{K}$ is taken over all $G$-equivariant operators such that $\tilde{K}\leq\frac{1}{n}\rho$. In particular, if $G$ acts irreducibly on the space $\cH$, then Schur's Lemma implies that 
\begin{align}q_*=\frac{d\lambda_{\min}(\rho)}{n}.\end{align}
In particular, $q_*>0$ if and only if $\rho$ has full rank. A choice of the optimal measurement is generated by a projection whose support is any subset of the minimal eigenspace of $\rho$.
Further studies of the state exclusion problem using representation-theoretic methods have recently been carried out in \cite{diebra2025exclusion,yao2025conclusive}.

\section{Generator with symmetry}
\label{sec:generators-symmetry}

In this section, we consider two examples where the representation is not irreducible, but the generator has a particular symmetry that reduces the problem into the irreducible case. Before stating the result, we will first review Schur-Weyl duality and quantum $k$-designs.

\subsection{Schur-Weyl Duality}
Consider the action of $S_n\times \cU_d$ on $(\mathbb{C}^d)^{\otimes n}$ where 
\begin{align}
	P_\pi(v_1\otimes...\otimes v_n)=v_{\pi^{-1}(1)}\otimes...\otimes v_{\pi^{-1}(n)}\quad\forall\:\pi\in S_n \label{eq:Sn-action}
\end{align}
and
\begin{align}
	U.(v_1\otimes...\otimes v_n)=U^{\otimes n}(v_1\otimes...\otimes v_n)=(Uv_1)\otimes...\otimes(Uv_n)\quad\forall\:U\in \cU_d. \label{eq:Ud-action}
\end{align}
Schur-Weyl decomposition states that the representation space $(\mathbb{C}^d)^{\otimes n}$ can be decomposed as follows \cite{fulton2013representation}:
\begin{align}
	(\mathbb{C}^d)^{\otimes n} \cong \bigoplus_{\lambda \vdash_d n} \cV_\lambda^d \otimes \cW_\lambda, \label{eq:schur-weyl-decomposition}
\end{align}
where $\lambda\vdash_d n$ denotes an unordered partition of $n$ into $d$ parts, $\cV_\lambda^d$ is an irreducible representation space of the unitary group $\cU_d$ with dimension $m_{d,\lambda}$, and $\cW_\lambda$ is an irreducible representation space of the symmetric group $S_n$ with dimension $d_\lambda$. In this basis, the representations \eqref{eq:Sn-action} and \eqref{eq:Ud-action} have the form
\begin{align} 
	P_\pi &\cong \bigoplus_{\lambda \vdash_d n} \one_{\cV_\lambda^d} \otimes p_\lambda(\pi) & \quad U^{\otimes n} &\cong \bigoplus_{\lambda \vdash_d n} q_\lambda(U) \otimes \one_{\cW_\lambda}, \label{eq:schur-weyl-action}
\end{align}
where $p_\lambda(\pi)$ and $q_\lambda(U)$ are the corresponding irreps of $S_n$ and $\cU_d$, respectively.
By Schur's lemma, any $S_n$-invariant operator can be written as 
\begin{align}X=\bigoplus_{\lambda\vdash_d n}X_\lambda\otimes\one_{\cW_\lambda},\end{align}
and similarly any $U^{\otimes n}$-invariant operator can be written as
\begin{align}X=\bigoplus_{\lambda\vdash_d n}\one_{\cV_\lambda^d}\otimes X_\lambda.\end{align}
In particular, any $(S_n\times\cU_d)$-invariant operator can be written as
\begin{align}X=\bigoplus_{\lambda\vdash_d n}c_\lambda\one_{\cV_\lambda^d}\otimes\one_{\cW_\lambda}\end{align}
where $c_\lambda\in\mathbb{C}$.

\subsection{Quantum designs} 
\begin{defn}
    A finite subset $X\subset\cU_d$ is called a $k$-design if 
    \begin{align}\frac{1}{|X|}\sum_{V\in X}V^{\otimes k}F(V^\dagger)^{\otimes k}=\int_{\cU_d}U^{\otimes k}F(U^\dagger)^{\otimes k}dU\quad\forall\:F\in\cB((\mathbb{C}^d)^{\otimes k}).\end{align}
\end{defn}
\noindent In particular, designs with group structure (called group designs) are of particular interest. It has been shown that the Pauli group is a 1-design but not a 2-design, and the Clifford group is a 3-design but not a 4-design \cite{webb2016design}. 
Most discussions about group designs use analytic methods; here, we present a representation-theoretic perspective that has been adopted in previous works such as \cite{gross2007evenly,kaposi2024generalizedgroupdesignsovercoming}:
\begin{prop}\label{Proposition: n design irrep}
    A group $G$ is an $n$-design on $\mathbb{C}^d\Longrightarrow G$ acts irreducibly on each $\cV_\lambda^d$ in \eqref{eq:schur-weyl-decomposition} for each $\lambda\vdash_d n$.
    %\felix{Are you sure you mean the Specht modules $\cW_\lambda$ and not the Weyl modules $\cV_\lambda^d$?}
\end{prop}
\begin{proof}
    According to \eqref{eq:schur-weyl-action} we have $g^{\otimes n}\cong \bigoplus_{\lambda} \phi_\lambda(g)\otimes \one_{\cW_\lambda}$ for some representations $\phi_\lambda$. Assume that there exists $\lambda\vdash_d n$ such that $\phi_\lambda$ is reducible, that is, there exists a non-trivial strict subspace $\lbrace 0\rbrace \lneqq \cT_\lambda\lneqq \cV_\lambda^d$ such that $\phi_\lambda(G)\cT_\lambda \subseteq \cT_\lambda$.
    Choose non-zero vectors $|v_\lambda\rangle\in\cT_\lambda$ and $|v_\lambda^\perp\rangle\in\cT_\lambda^\perp$.
    By assumption, these vectors satisfy $\langle v_\lambda^\perp|\phi_\lambda(g)|v_\lambda\rangle = 0$ for all $g\in G$.
    Choosing another non-zero vector $|w_\lambda\rangle\in\cW_\lambda$ and setting $X_\lambda = |v_\lambda\rangle\langle v_\lambda|\otimes |w_\lambda\rangle\langle w_\lambda|$, we have
    \begin{align}
    	0 = \left(\frac{1}{|G|} \sum\nolimits_{g\in G} g^{\otimes n} X_\lambda (g^\dagger)^{\otimes n}\right) |v_\lambda^\perp\rangle \otimes |w_\lambda\rangle = \left(\int_{\cU_d} U^{\otimes n} X_\lambda (U^\dagger)^{\otimes n} dU\right) |v_\lambda^\perp\rangle \otimes |w_\lambda\rangle = c_\lambda |v_\lambda^\perp\rangle \otimes |w_\lambda\rangle
    \end{align}
	for some $c_\lambda\neq 0$, which is a contradiction.
	Hence, $\phi_\lambda$ is irreducible for each $\lambda\vdash_d n$.
\end{proof}

\subsection{Werner generator state}\label{sec:werner-generator-state}
Let $\rho_n$ be a $U^{\otimes n}$-invariant state for all $U\in\cU_d$, then
\begin{align}
	\rho_n = \bigoplus_{\lambda\vdash_d n} \one_{V_\lambda} \otimes \rho_\lambda,
\end{align}
for some positive operators $\rho_\lambda$ on $W_\lambda$.

We define a geometrically uniform ensemble $(\rho^{\pi})_{\pi\in S_n}$ by
\begin{align}\rho^{\pi} = P_\pi \rho_n P_\pi^\dagger \cong \bigoplus_{\lambda \vdash_d n} \one_{V_\lambda} \otimes \rho_{\pi,\lambda}\end{align}
 where \begin{align}\rho_{\pi,\lambda}=p_\lambda(\pi) \rho_\lambda p_\lambda(\pi)^\dagger.\end{align}
 Note that $p_\lambda$ is an irreducible representation of $S_n$ on $W_\lambda$ for each $\lambda$, so \Cref{prop:geom uniform optimal} can be applied to each $W_\lambda$ individually. So we have proved:

\begin{prop}
\label{prop:werner-generator}
    Let $\rho_n=\bigoplus_\lambda c_\lambda\one_{V_\lambda} \otimes \rho_\lambda$ be a $U^{\otimes n}$-invariant state and $\rho^\pi\coloneqq P_\pi\rho_n P_\pi^\dagger$ form a geometrically uniform ensemble $(\rho^{\pi})_{\pi\in S_n}$. Then the optimal success probability of state discrimination is
    \begin{align}\psucc^*=\frac{1}{n!}\sum_\lambda d_\lambda m_{d,\lambda}\lmax(\rho_\lambda).\end{align}
    An optimal choice of POVM is given by
    \begin{align}M_\pi=P_\pi\bigg(\bigoplus_\lambda\frac{m_{d,\lambda}}{d^{\max}_\lambda}\one_{V_\lambda}\otimes\Pi^{\max}_{\rho_\lambda}\bigg)P_\pi^\dagger\end{align}
    where $d_\lambda^{\max}$ is the dimension of the maximal eigenspace of $\rho_\lambda$.
\end{prop}

\subsection{Permutation-invariant generator state}\label{sec: perm invariant generator}

Consider a permutation-invariant state $\rho_n$ on $(\mathbb{C}^d)^{\otimes n}$, which can be written as $\rho_n = \bigoplus_\lambda \rho_\lambda \otimes \one_{\cW_\lambda}$ for some positive operators $\rho_\lambda$ on $\cV_\lambda^d$.
Let the group $G$ be an $n$-design on $\mathbb{C}^d$. 
We define the geometrically uniform ensemble $(\rho_g)_{g\in G}$, where
\begin{align}
	\rho_g = g^{\otimes n} \rho_n (g^\dagger)^{\otimes n}=\bigoplus_\lambda q_\lambda(g)\rho_\lambda q_\lambda(g)^\dagger\otimes\one_{W_\lambda}.
\end{align}
Because of \Cref{Proposition: n design irrep}, \Cref{prop:geom uniform optimal} can be applied to each $V_\lambda$ individually. So we have proved:
\begin{prop}
\label{prop:perm-inv-generator}
    Let $G$ be a finite group forming an $n$-design on $\mathbb{C}^d$, $\rho_n=\bigoplus_\lambda \rho_\lambda\otimes \one_{W_\lambda}$ be an $S_n$-invariant state and $\rho_g\coloneqq g^{\otimes n}\rho_n (g^{\otimes n})^\dagger$ form a geometrically uniform ensemble $(\rho_g)_{g\in G}$. Then the optimal success probability of state discrimination is
    \begin{align}\psucc^*=\frac{1}{|G|}\sum_\lambda d_\lambda m_{d,\lambda}\lmax(\rho_\lambda).\end{align} 
    An optimal choice of POVM is given by
    \begin{align}M_g=g^{\otimes n}\bigg(\bigoplus_\lambda \frac{d_\lambda}{d_\lambda^{\max}}\Pi^{\max}_{\rho_\lambda}\otimes\one_{W_\lambda}\bigg)(g^\dagger)^{\otimes n},
    \end{align}
	where $d_\lambda^{\max}$ is again the dimension of the maximal eigenspace of $\rho_\lambda$.
\end{prop}

\section{Pure GU ensembles}\label{sec:pureGU}
In this section, we do not assume irreducibility but assume that the generator state $\rho$ of the geometrically uniform ensemble is pure, $\rho=\ket{\psi}\bra{\psi}$. 
In this situation, \textcite{eldar2004optimal} proved that the pretty good measurement is the optimal measurement in discriminating the states in the GU ensemble, and \textcite{dallapozza2015optimality} gave necessary and sufficient conditions for the optimality of the pretty good measurement for \emph{compound} pure-state GU ensembles.
We will reprove the optimality of the pretty good measurement for pure-state GU ensembles from \cite{eldar2004optimal} by directly computing the success probability $p_{\mathrm{PGM}}$ and proving its optimality.
Here we also note the previous work \cite{krovi2015optimal} that derived a representation-theoretic formula for the optimal success probability of discriminating pure-state GU ensembles.

We first note that equation \eqref{eq: ppgm} for the success probability using the pretty good measurement can be simplified to the following expression:
\begin{align}
    p_{\mathrm{PGM}}=\frac{1}{n}\tr(\overline{\rho}^{-\frac{1}{2}}\rho\overline{\rho}^{-\frac{1}{2}}\rho)=\frac{1}{n}\bra{\psi}\overline{\rho}^{-\frac{1}{2}}\ket{\psi}^2 = \frac{1}{n} \left[\tr(\overline{\rho}^{-\frac{1}{2}} |\psi\rangle\langle\psi|)\right]^2.
\end{align}

By construction the average state $\overline{\rho}$ commutes with all $g\in G$, and hence also $g^\dagger \overline{\rho}^{-\frac{1}{2}} g = \overline{\rho}^{-\frac{1}{2}}$ for all $g\in G$. 
It follows that $\tr(\overline{\rho}^{-\frac{1}{2}} |\psi\rangle\langle\psi|) = \tr(g^\dagger\overline{\rho}^{-\frac{1}{2}} g|\psi\rangle\langle\psi|) = \tr(\overline{\rho}^{-\frac{1}{2}} g|\psi\rangle\langle\psi|g^\dagger)$ for all $g\in G$, and averaging over the full group $G$ gives
\begin{align}
	\tr(\overline{\rho}^{-\frac{1}{2}} |\psi\rangle\langle\psi|) = \frac{1}{n} \sum_{g\in G} \tr(\overline{\rho}^{-\frac{1}{2}} g|\psi\rangle\langle\psi|g^\dagger) =  \tr(\overline{\rho}^{-\frac{1}{2}}\overline{\rho}) = \tr(\overline{\rho}^{\frac{1}{2}}).
\end{align}
We have thus proved:
\begin{lem}
\label{lem:GU-PGM}
    Let $G$ be a finite group with $n=|G|$, and let $(n^{-1},\psi_g)_{g\in G}$ be a pure-state GU ensemble with average state $\overline{\rho}= \frac{1}{n}\sum_{g\in G} \psi_g$.
    Then the success probability of the pretty-good measurement is equal to 
    \begin{align}
	p_{\mathrm{PGM}} = \frac{1}{n}\left[\tr(\overline{\rho}^{\frac{1}{2}})\right]^2.
	\label{eq:pureGUppgm}
\end{align}
\end{lem}

Let us recall the following lower and upper ``generalized Holevo-Curlander'' bounds on the optimal success probability $\psucc^*$ of a \emph{generic} quantum state ensemble $(p_i,\rho_i)_{i\in [n]}$ \cite{ogawa1999strong,tyson2009two-sided,tyson2010two-sided}:
\begin{align}
    \left[ \tr \left( \sum\nolimits_{i=1}^n p_i^2 \rho_i^2\right)^{1/2} \right]^2 \leq \psucc^* \leq \tr \left( \sum\nolimits_{i=1}^n p_i^2 \rho_i^2\right)^{1/2}
    \label{eq:holevo-curlander}
\end{align}
It is straightforward to see that for a uniformly weighted pure-state ensemble the left-most expression in \eqref{eq:holevo-curlander} is equal to the expression in \eqref{eq:pureGUppgm}, and hence the lower Holevo-Curlander bound on $\psucc^*$ coincides with $p_{\mathrm{PGM}}$ for pure-state GU ensembles.

We can also compare Lemma \ref{lem:GU-PGM} to the following result proved for a generic pure state ensemble $(p_i,|\psi_i\rangle)_i$ in \cite{montanaro2007distinguishability}:
    \begin{align}
    	\ppgm=\sum_i(\sqrt{G})_{ii}^2\geq\frac{1}{n}\left[\tr(G^\frac{1}{2})\right]^2,
    \end{align}
    where $G$ is the Gram matrix of the unnormalized, weighted ensemble with $(G)_{ij} = \langle \varphi_i|\varphi_j\rangle$ and $|\varphi_i\rangle = \sqrt{p_i}|\psi_i\rangle$. Since $G$ has the same spectrum as $\overline{\rho}$, \Cref{lem:GU-PGM} shows that this lower bound is always achieved for a pure state GU ensemble.

Lemma \ref{lem:GU-PGM} yields a new and simple proof of the fact that the PGM is optimal for discriminating a pure-state GU ensemble, which is a special case of the result in \cite{eldar2004optimal}:

\begin{prop}[\cite{eldar2004optimal}]
\label{prop:pure-GU-PGM-optimal}
    For a GU ensemble with pure states, the PGM is optimal.
\end{prop}
\begin{proof}
Consider the operator
\begin{align}
	K=\frac{1}{n}\tr(\overline{\rho}^{\frac{1}{2}})\overline{\rho}^{\frac{1}{2}}.
\end{align} 
To prove the claim, it suffices to show that $K$ satisfies $K\geq\frac{1}{n}\rho$, since then $K$ is a feasible point in the dual program \eqref{eq:dual} and 
\begin{align}
	p_{\mathrm{PGM}} \leq p^*\leq\tr(K)=\frac{1}{n}\tr(\overline{\rho}^{\frac{1}{2}})^2=p_{\mathrm{PGM}}.
\end{align}
To this end, consider the rank-1 operator
\begin{align}
	P\coloneqq\overline{\rho}^{-\frac{1}{4}}\rho\overline{\rho}^{-\frac{1}{4}}=\overline{\rho}^{-\frac{1}{4}}\ket{\psi}\bra{\psi}\overline{\rho}^{-\frac{1}{4}},
\end{align}
whose non-zero eigenvalue (corresponding to the eigenvector $\overline{\rho}^{-\frac{1}{4}}|\psi\rangle$) is equal to
\begin{align}\bra{\psi}\overline{\rho}^{-\frac{1}{2}}\ket{\psi}=\tr(\overline{\rho}^{\frac{1}{2}}),\end{align}
so that $P\leq \tr(\overline{\rho}^{\frac{1}{2}}) \one$.
Hence,
\begin{align}\frac{1}{n}\rho=\frac{1}{n}\overline{\rho}^{\frac{1}{4}}P\overline{\rho}^{\frac{1}{4}}\leq\frac{1}{n}\tr(\overline{\rho}^{\frac{1}{2}})\overline{\rho}^{\frac{1}{2}}=K,
\end{align}
proving feasibility of $K$ in \eqref{eq:dual} and thus concluding the proof.
\end{proof}

\begin{rem}
    \Cref{prop:pure-GU-PGM-optimal} together with \Cref{ex:coupon} show that the lower Holevo-Curlander bound on $\psucc^*$ in \eqref{eq:holevo-curlander} is in fact tight and equal to the PGM success probability \eqref{eq:pureGUppgm}.
    These bounds were used in \cite{hadiashar2024optimal} to give a lower bound on the sample complexity in the quantum coupon collector problem (see \cite{arunachalam2020quantumcoupon} and \Cref{ex:coupon}).
\end{rem}

\subsection{Example: Hidden subgroup problem over semidirect product groups}
The hidden subgroup problem over groups of the form $G=A\rtimes_\varphi\mathbb{Z}_p$ \cite{bacon2005semidirect}, where $A$ is abelian and $p$ is prime, can be reduced to the discrimination of the following uniform ensemble for $d\in A$:
\begin{align}
	\rho_d^{\otimes k}=\frac{1}{|G|^k}\bigoplus_{x\in A^k}\ket{\phi_{x,d}}\bra{\phi_{x,d}}\otimes\ket{x}\bra{x}
	\label{eq:HSG-semidirect}
\end{align}
where
\begin{align}
	\ket{\phi_{x,d}}\coloneqq\bigotimes_{i=1}^{k}\left(\sum\nolimits_{b_i\in\mathbb{Z}_p}\chi_{x_i}\left(\Phi^{(b)}(d)\right)\ket{b}\right)
\end{align}
with $\Phi^{(b)}(d)\coloneqq\sum_{i=0}^{b-1}\varphi^i(d)$ and $\chi_a$ the $a$-th irreducible group character of $A$. Note that $\ket{\phi_{x,d}}$ is a pure state GU ensemble generated by the group $A$ with the action $a\colon\ket{\phi_{x,d}}\mapsto\ket{\phi_{x,ad}}$, so the PGM is optimal by \Cref{prop:pure-GU-PGM-optimal} (and \cite{eldar2004optimal}), and by \eqref{eq:pureGUppgm} it suffices to compute $\tr(\overline{\phi}_x^{\frac{1}{2}})$ for all $x\in A^k$ where $\overline{\phi_x}=\frac{1}{|A|}\sum_{d\in A}\ket{\phi_{x,d}}\bra{\phi_{x,d}}$.

By \cite[Lemma 7]{bacon2005semidirect}, for any $b_i\in\mathbb{Z}_p$ there exists a function $\hat{\Phi}^{(b_i)}\colon A\to A$ such that for all $a,d\in A$ we have
\begin{align}
	\chi_a(\Phi^{(b_i)}(d))=\chi_{\hat{\Phi}^{(b_i)}(a)}(d).
\end{align} 
Let $S_w^x:=\{b\in(\mathbb{Z}_p)^k:\hat{\Phi}^{(b)}(x)=w\}$ and $\ket{S_w^x}:=\frac{1}{\sqrt{|S_w^x|}}\sum_{b\in S_w^x}\ket{b}$, then the average is explicitly computed to be
\begin{align}
    \overline{\phi_x}=\sum_{w\in A}|S_w^x|\ket{S_w^x}\bra{S_w^x}.
\end{align}
Therefore by \cref{eq:pureGUppgm} the optimal success probability is 
\begin{align}
	p^*=\frac{1}{|G|^k}\sum_{x\in A^k}\frac{1}{|A|}\tr\left(\overline{\phi}_x^{\frac{1}{2}}\right)^2=\frac{p}{|G|^{k+1}}\sum_{x\in A^k}(\sum_{w\in A}\sqrt{|S_w^x|})^2.
\end{align}
This recovers the result in \cite{bacon2005semidirect}.

\subsection{Example: Port-based teleportation entanglement fidelity} 
\label{sec:pbt}

\subsubsection{Set up}
In port-based teleportation the entanglement fidelity of the teleportation channel reduces to a state discrimination problem \cite{ishizaka2008asymptotic,ishizaka2009quantum,Beigi_2011,chitambar2024teleportation}:
\begin{align} 
F(\Lambda)=\frac{n}{d^2}\psucc
\end{align}
where the ensemble in the state discrimination is given by
\begin{align} 
\rho_i=\tr_{B_i^c}(\phi_{AB}^+)^{\otimes n}=\phi^+_{A_iB}\otimes \frac{1}{d^{n-1}}\one_{A_i^c}.
\label{eq:pbt-states}
\end{align}
This is a GU ensemble generated by $S_n$ acting on the $A^n$-system of $\rho=\rho_1$.
Note that in this ensemble each state is repeated $(n-1)!$ times. Therefore, if $p^*$ is the optimal success probability of discriminating the GU ensemble generated by $\rho=\rho_1$ and $S_n$, then
\begin{align} 
F(\Lambda)=\frac{n(n-1)!}{d^2}p^*=\frac{n!}{d^2}p^*.
\label{eq:fidelity-psucc}
\end{align}

%\felix{Add short footnote explaining the repetition and how it's not a big issue.}

\subsubsection{Decomposition of the space}
Following \cite{christandl2021asymptotic,leditzky2020}, we first note that by the dual Pieri rule the representation space $(\mathbb{C}^d)^{\otimes (n+1)}$ decomposes as a $(U(d)\times S_n)$-representation:
\begin{align}\label{eq: space decomp 1}
    (\mathbb{C}^d)^{\otimes (n+1)}\cong\bigoplus_{\alpha\vdash_d(n-1)}V_\alpha\otimes\left(\bigoplus_{\mu=\alpha+\square}W_\mu\right).
\end{align}
Restricting $W_\mu$ to $S_{n-1}$ gives
\begin{align} 
W_\mu\big\downarrow_{S_{n-1}}\cong\bigoplus_{\beta=\mu-\square}W_\beta,
\end{align}
and therefore we can decompose $(\mathbb{C}^d)^{\otimes (n+1)}$ as a $(U(d)\times S_{n-1})$-representation as follows:
\begin{align}
    (\mathbb{C}^d)^{\otimes (n+1)}\cong\bigoplus_{\alpha\vdash_d(n-1)}V_\alpha\otimes\left(\bigoplus_{\mu=\alpha+\square}\bigoplus_{\beta=\mu-\square}W_\beta\right).
\end{align}
Grouping by $\beta=\alpha$ and $\beta\neq\alpha$, we get
\begin{align}\label{eq: space decomp 2}
    (\mathbb{C}^d)^{\otimes (n+1)}\cong\bigoplus_{\alpha\vdash_d(n-1)}V_\alpha\otimes\left(W_\alpha\otimes\mathbb{C}^{n_\alpha}\oplus \left(\bigoplus_{\mu=\alpha+\square}\bigoplus_{\substack{\beta=\mu-\square\\ \beta\neq\alpha}}W_\beta \right)\right)
\end{align}
where $n_\alpha=\#\{\mu:\mu=\alpha+\square\}$.

There is a more direct way to derive \cref{eq: space decomp 2}: consider the $(U(d)\times S_{n-1})$-action on $A_2...A_n$ and $U(d)$-action on $A_1B$ respectively:
\begin{align}\label{eq: space decomp 3}
    (\mathbb{C}^d)^{\otimes(n+1)}=H_{A_1B}\otimes H_{A_2...A_n}\cong\bigg(V_{\square}\otimes V_{\square}^*\bigg)\otimes\bigg(\bigoplus_{\alpha\vdash_d(n-1)}V_\alpha\otimes W_\alpha\bigg)
\end{align}
Since $V_\square\otimes V_\square^*\otimes V_\alpha\cong\bigoplus_{\mu=\alpha+\square}\bigoplus_{\beta=\mu-\square}V_\beta$, we then get
\begin{align}\label{eq: space decomp 4}
    (\mathbb{C}^d)^{\otimes(n+1)}\cong\bigoplus_{\alpha\vdash_d(n-1)}\bigoplus_{\mu=\alpha+\square}\bigoplus_{\beta=\mu-\square}V_\beta\otimes W_\alpha\cong\bigoplus_{\alpha\vdash_d(n-1)}V_\alpha\otimes\bigg(W_\alpha\otimes\mathbb{C}^{n_\alpha}\oplus(\bigoplus_{\mu=\alpha+\square}\bigoplus_{\substack{\beta=\mu-\square\\ \beta\neq\alpha}}W_\beta)\bigg).
\end{align}

\subsubsection{Decomposition of $\rho$}
Consider now the following state from \eqref{eq:pbt-states} that generates the ensemble in the state discrimination problem associated with port-based teleportation:
\begin{align} 
\rho=\rho_1=\ket{\phi^+}\bra{\phi^+}_{A_1B}\otimes \frac{1}{d^{n-1}}\one_{A_2...A_n},
\end{align}
where $\ket{\phi^+}_{A_1B}=\frac{1}{\sqrt{d}}\sum_{i=1}^{d}\ket{i}_{A_1}\ket{i}_B$. Since $\rho$ is $U^{\otimes n}\otimes\overline{U}$ invariant, $\rho$ has the following form with respect to the decomposition \cref{eq: space decomp 1}:
\begin{equation}\label{eq: decomp of rho}
    \rho=\bigoplus_{\alpha\vdash_d(n-1)}\frac{1}{m_\alpha}\one_{V_\alpha}\otimes\rho_\alpha.
\end{equation}
Now consider \cref{eq: space decomp 3} and \cref{eq: space decomp 4}.
Under the canonical isomorphism $V_{\square}\otimes V_{\square}^*\cong\text{End}(V_{\square})$, the group $U(d)$ acts on $\text{End}(V_{\square})$ by conjugation, and the only fixed point of this action is the identity matrix up to scalars. 
On the other hand, $\ket{\phi^+}_{A_1B}\in V_{\square}\otimes V_{\square}^*$ corresponds to a multiple of the identity matrix in $\text{End}(V_{\square})$ under this isomorphism, and hence $\ket{\phi^+}_{A_1B}$ is a vector in the trivial representation $V_{\emptyset}\subset V_{\square}\otimes V_{\square}^*$. Note that $V_\emptyset\otimes V_\alpha\cong V_\alpha$, so $\ket{\phi^+}_{A_1B}$ corresponds to a vector in $\mathbb{C}^{n_\alpha}$. As a result, 
\begin{align}\label{eq: PBT decomposition}
    \rho=\bigoplus_{\alpha\vdash_d(n-1)}\frac{1}{m_\alpha}\one_{V_\alpha}\otimes\bigg(\frac{1}{d_\alpha}\one_{W_\alpha}\otimes\psi_\alpha\oplus(0)_\alpha\bigg),
\end{align}
where $\psi_\alpha$ is a rank-one operator acting on $\mathbb{C}^{n_\alpha}$, and $(0)_\alpha$ denotes a fixed state supported in the kernel of $\rho$. Comparing with \cref{eq: decomp of rho} we get, with respect to the decomposition in \cref{eq: space decomp 4},
\begin{align} \label{eq: PBT generator}
\rho_\alpha=\frac{1}{d_\alpha}\one_{W_\alpha}\otimes\psi_\alpha\oplus(0)_\alpha.
\end{align}
\begin{rem}
    Note that \cref{eq: PBT generator} is technically not a pure-state GU generator because $W_\alpha\otimes\mathbb{C}^{n_\alpha}$ is not $S_n-$invariant, so we cannot directly cite \Cref{lem:GU-PGM} and \Cref{prop:pure-GU-PGM-optimal}. Rather, it is a direct generalization from the operator algebraic perspective, and the idea is essentially the same as we show in the next subsection.
\end{rem}

\subsubsection{Success probability of PGM}
With respect to the decomposition in \cref{eq: space decomp 1}, the average state is given by 
\begin{align}
    \overline{\rho}=\bigoplus_{\alpha\vdash_d(n-1)}\bigoplus_{\mu=\alpha+\square}r_{\alpha,\mu}\one_{V_\alpha}\otimes\one_{W_\mu}.
\end{align}
where $r_{\alpha,\mu}=\frac{1}{d^n}\frac{m_\mu d_\alpha}{m_\alpha d_\mu}$ \cite{leditzky2020}. Comparing this with \cref{eq: decomp of rho}, the average of $\rho_\alpha$ can be expressed with respect to \cref{eq: space decomp 4} as
\begin{align}
    \overline{\rho}_\alpha=\left(\frac{1}{d_\alpha}\one_{W_\alpha}\otimes D_\alpha\right)\oplus Q
\end{align}
where $D_\alpha$ is a diagonal matrix acting on $\mathbb{C}^{n_\alpha}$ with $(D_\alpha)_{\mu\mu}=\frac{d_\alpha^2 m_\mu}{d^n d_\mu}$ and $\supp Q\perp\supp\rho_\alpha$. Therefore,
\begin{align}
    \ppgm=\frac{1}{n!}\tr(\overline{\rho}^{-\frac{1}{2}}\rho\overline{\rho}^{-\frac{1}{2}}\rho)=\frac{1}{n!}\sum_{\alpha\vdash_d(n-1)}\bra{\psi_\alpha}D_\alpha^{-\frac{1}{2}}\ket{\psi_\alpha}^2.
\end{align}
Since $D_\alpha$ is diagonal,
\begin{align}
    \bra{\psi_\alpha}D_\alpha^{-\frac{1}{2}}\ket{\psi_\alpha}=\tr(D_\alpha^{-\frac{1}{2}}\psi_\alpha)=\sum_{\mu=\alpha+\square}\frac{(\psi_\alpha)_{\mu\mu}}{\sqrt{(D_\alpha)_{\mu\mu}}}.
\end{align}
Note that $(D_\alpha)_{\mu\mu}=\frac{d_\alpha}{d_\mu}(\psi_\alpha)_{\mu\mu}$ because $\overline{\rho}_\alpha=\bigoplus_{\mu=\alpha+\square}\tr\bigg((\rho_\alpha)_{\mu\mu}\bigg)\frac{1}{d_\mu}\one_{W_\mu}$.
Hence,
\begin{align}
    \bra{\psi_\alpha}D_\alpha^{-\frac{1}{2}}\ket{\psi_\alpha}=\sum_{\mu=\alpha+\square}\sqrt{\frac{d_\mu}{d_\alpha}(\psi_\alpha)_{\mu\mu}}=\frac{1}{\sqrt{d^n}}\sum_{\mu=\alpha+\square}\sqrt{m_\mu d_\mu},
\end{align}
and therefore
\begin{align}
    \ppgm=\frac{1}{n!d^n}\sum_{\alpha\vdash_d(n-1)}(\sum_{\mu=\alpha+\square}\sqrt{m_\mu d_\mu})^2.\label{eq:pbt-pgm}
\end{align}
Substituting this expression in \eqref{eq:fidelity-psucc} now yields the known exact expression for the entanglement fidelity of the standard PBT protocol operating on $n$ maximally entangled states first derived in \cite{studzinski2017port} (see also \cite{leditzky2020}).

\subsubsection{Fully optimized protocol}
In the fully optimized PBT setting one aims to maximize the entanglement fidelity of the teleportation channel by choosing the optimal resource state.
It is known that pairs of bipartite maximally entangled states are in fact sub-optimal, and without loss of generality the resource state can be chosen as a purification of a symmetric Werner state \cite{ishizaka2009quantum,mozrzymas2018optimal,christandl2021asymptotic}.

In the state discrimination picture, this amounts to a state ensemble consisting of quantum states
$\tilde{\rho}_i=(X_{A^n}\otimes \one_B)\rho_i(X_{A^n}^\dagger\otimes \one_B)$ where $X_{A^n}$ has both $U^{\otimes n}$ and $S_n$ symmetry, i.e.
\begin{align} 
X_{A^n}=\bigoplus_{\mu\vdash_d n} \sqrt{c_\mu}\,\one_{V_\mu}\otimes \one_{W_\mu}
\label{eq:X-decomposition1}
\end{align}
for some coefficients $c_\mu$. Since $V_\mu\otimes\mathcal{H}_B\cong\bigoplus_{\alpha+\square=\mu}V_\alpha$ as a $U(d)-$representation, we have
\begin{align}
X_{A^n}\otimes \one_B=\bigoplus_{\mu\vdash_d n}\sqrt{c_\mu}\one_{\bigoplus_{\alpha+\square=\mu}V_\alpha}\otimes\one_{W_\mu}=\bigoplus_{\alpha\vdash_d(n-1)}\bigoplus_{\mu=\alpha+\square}\sqrt{c_\mu}\,\one_{V_\alpha}\otimes\one_{W_\mu},
\label{eq:X-decomposition2}
\end{align}
therefore 
\begin{align}
    (\tilde{\psi}_\alpha)_{\mu\mu}=c_\mu(\psi_\alpha)_{\mu\mu}.
\end{align}
The same calculation procedure as in the previous section then gives
\begin{align}
    \ppgm=\frac{1}{n!d^n}\sum_{\alpha\vdash_d(n-1)}(\sum_{\mu=\alpha+\square}\sqrt{c_\mu m_\mu d_\mu})^2,
\end{align}
which substituted into \eqref{eq:fidelity-psucc} yields the exact expression of the entanglement fidelity of the fully optimal PBT protocol first derived in \cite{mozrzymas2018optimal} (see also \cite{leditzky2020}).

\subsubsection{Optimality of PGM} 
Using our approach we also provide a new proof of the optimality of the pretty good measurement for PBT \cite{mozrzymas2018optimal,leditzky2020}.
By \Cref{prop:block-diagonal}, it suffices to show that PGM is the optimal measurement for the (normalized) generator states $\rho_\alpha$ for each $\alpha$. We first review the condition of the optimality of PGM given by Eldar:
\begin{lem}[\cite{eldar2004optimal}]\label{lem: PGM optimal Eldar}
    If $\rho=\phi\phi^*$ and $\phi^*\overline{\rho}^{-\frac{1}{2}}\phi=c\one$ for some scalar $c$, then PGM is optimal for the GU ensemble with generator $\rho$ and average state $\overline{\rho}$.
\end{lem}
In particular, let $\phi_\alpha=\frac{1}{\sqrt{d_\alpha}}\one_{W_\alpha}\otimes\ket{\psi_\alpha}$, then $\rho_\alpha=\phi_\alpha\phi_\alpha^*$ and
\begin{align}
    \phi_\alpha^*\overline{\rho}_\alpha^{-\frac{1}{2}}\phi_\alpha=\frac{\bra{\psi_\alpha}D_\alpha^{-\frac{1}{2}}\ket{\psi_\alpha}}{\sqrt{d_\alpha}}\one_{W_\alpha},
\end{align}
so PGM is optimal.

\section{General bounds on PGM}
\label{sec:general-GU}
\subsection{A general lower bound on PGM for generic GU ensembles}
In this section we let $G\ni g\mapsto U_g$ be an arbitrary (projective) representation of a finite group $G$.\footnote{Phases of the form $U_gU_h = c(g,h) U_{gh}$ with $|c(g,h)|=1$ do not matter in the GU state discrimination problem.}
Let $\sigma$ be a generator state on $\mathbb{C}^d$ and for $n\in\mathbb{N}$ consider the $n$-copy GU ensemble $\cE=(q_g,\sigma_g^{\otimes n})_{g\in G}$ with $q_g = \frac{1}{|G|}$ and $\sigma_g = U_g\sigma U_g^\dagger$.
We aim to derive a lower bound on the success probability $\ppgm(\cE)$ of distinguishing the states in $\cE$ with the PGM.
To this end, we first pass to the pure-state ensemble obtained from the spectral decompositions of the mixed states $\sigma_g^{\otimes n}$, following~\cite{montanaro2007distinguishability}.

Let $\sigma = \sum_{i=0}^{d-1} \lambda_i |i\rangle\langle i|$ be a spectral decomposition of the generator state.
Then $\sigma_g = \sum_i \lambda_i |i;g\rangle\langle i;g|$ with $|i;g\rangle\coloneqq U_g |i\rangle$ is a spectral decomposition of $\sigma_g$ for $g\in G$.
In the $n$-copy setting, we have spectral decompositions
\begin{align}
	\sigma_g^{\otimes n} = \sum_{i^n\in [d]^n} \lambda_{i^n} |i^n; g\rangle \langle i^n;g|,
\end{align} 
where we use the following notation: $[d]^n$ is the set of all $d$-ary strings of length $n$, and for $i^n = (i_1,\dots,i_n)$ (with $i_j\in [d]$) we set
\begin{align}
	\lambda_{i^n} &\coloneqq \prod_{k=1}^n \lambda_{i_k} & 
	|i^n\rangle &\coloneqq \bigotimes_{k=1}^n |i_k\rangle &
	|i^n;g\rangle &\coloneqq U_g^{\otimes n} |i^n\rangle.
\end{align}
Finally, we define the pure-state ensemble $\cF= (q_{i^n,g}, |i^n;g\rangle)_{i^n\in [d]^n, g\in G}$ with $q_{i^n,g} = \lambda_{i^n}/|G|$.
For this ensemble, $\ppgm(\cE) \geq \ppgm(\cF)$ by \cite[Lemma 2.4]{montanaro2007distinguishability}.

To bound $\ppgm(\cF)$ from below, we use the following lower bound based on pair-wise fidelities proved in \cite[eq.~(9)]{montanaro2007distinguishability}:
\begin{align}
	\ppgm(\cF) \geq \sum_{i^n\in [d]^n} \sum_{g\in G} q_{i^n,g}^2 \left(\sum\nolimits_{j^n,h} q_{j^n,h} |\langle i^n;g| j^n;h\rangle|^2\right)^{-1}\label{eq:pairwise-bound}
\end{align}
Fixing $i^n\in [d]^n$ and $g\in G$, we compute:
\begin{align}
	\sum_{j^n,h} q_{j^n,h} |\langle i^n;g| j^n;h\rangle|^2 &= \frac{1}{|G|}\sum_{j^n,h} \lambda_{j^n} |\langle i^n| (U_g^\dagger)^{\otimes n} U_h^{\otimes n}| j^n\rangle|^2 \\
	&= \frac{1}{|G|}\sum_{j^n,h} \lambda_{j^n} |\langle i^n| U_{h}^{\otimes n} | j^n\rangle|^2\label{eq:group-bijection}\\
	&= \frac{1}{|G|}\sum_{j^n,h} \lambda_{j^n} |\langle i^n| U_{h}^{\otimes n} | j^n\rangle\langle j^n | (U_h^\dagger)^{\otimes n} | i^n\rangle\\
	&= \frac{1}{|G|} \sum_h   \langle i^n| U_{h}^{\otimes n} \left(\sum\nolimits_{j^n}\lambda_{j^n}| j^n\rangle\langle j^n |\right) (U_h^\dagger)^{\otimes n} | i^n\rangle\\
	&= \frac{1}{|G|} \sum_h   \langle i^n| U_{h}^{\otimes n} \sigma^{\otimes n} (U_h^\dagger)^{\otimes n} | i^n\rangle\\
	&= \frac{1}{|G|} \sum_h  \prod_{k=1}^n \langle i_k|\sigma_h|i_k\rangle.\label{eq:denominator}
\end{align}
For \eqref{eq:group-bijection} note that $U_{g}^\dagger U_h = c(g^{-1},h) U_{g^{-1}h} = c(g^{-1},h) U_{h'}$ for some $h'\in G$. 
Since the phase $c(g^{-1},h)$ is irrelevant and $h\mapsto g^{-1}h$ is a bijection on $G$, the statement follows.

Using \eqref{eq:denominator} in \eqref{eq:pairwise-bound} and canceling factors of $|G|$, we get
\begin{align}
	\ppgm(\cF) \geq \sum_{i^n\in[d]} \frac{\lambda_{i^n}^2}{\sum_h  \prod_{k=1}^n \langle i_k|\sigma_h|i_k\rangle}.\label{eq:pairwise-compact}
\end{align}

\subsection{Sample complexity for the hidden subgroup problem}
The method in \cite{montanaro2007distinguishability} is powerful and not restricted to GU ensembles. Here we present an example of a non-GU ensemble (but still with group structures involved) where the method still applies, reproving an earlier result of \textcite{hayashi2006hiddensubgroup} without appealing to the Hayashi-Nagaoka operator inequality:
\begin{prop}[\cite{hayashi2006hiddensubgroup}]
    \label{prop:hayashi-bound}
    Let $\{H^i:i=1,...,m\}$ be a set of candidate subgroups of a finite group $G$ such that $|H^i|=|H|$ for all $i$. Then the sample complexity of solving the hidden subgroup problem is at most $O(\frac{\log m}{\log t})$ where $t\coloneqq \min_{1\leq i<j\leq m}\frac{|H|}{|H^i\cap H^j|}$.
\end{prop}
\begin{proof}
    Define 
\begin{align}
	\rho_i\coloneqq\frac{1}{|K|}\sum_{c\in K_i}\ket{cH^i}\bra{cH^i}
\end{align}
where $K_i$ is a set of coset representatives of $H^i$ and $\ket{cH^i}\coloneqq|H|^{-1/2}\sum_{h\in H^i}\ket{ch}$. Then 
\begin{align}
	\rho_i^{\otimes n}=\frac{1}{|K|^n}\sum_{c^n\in K_i^n}\ket{c^nH^i}\bra{c^nH^i}
\end{align}
where $\ket{c^nH^i}=\bigotimes_{l=1}^{n}\ket{c_lH^i}$. 
Now apply \cite[eq.~(10)]{montanaro2007distinguishability} to the uniformly distributed pure state ensemble $\{\ket{c^nH^i}:1\leq i\leq m, c^n\in K_i^n\}$ to get
\begin{align}
	\ppgm\geq\frac{1}{m|K|^n}\sum_{1\leq i\leq m,c^n\in K_i^n} \left(\sum_{1\leq j\leq m,d^n\in K_j^n}|\langle c^nH^i | d^nH^j\rangle|^2\right)^{-1}.
\end{align}
For a fixed $(i,c^n)$,
\begin{align}
    \sum_{1\leq j\leq m,d^n\in K_j^n}|\langle c^nH^i | d^nH^j\rangle|^2&=\sum_{1\leq j\leq m,d^n\in K_j^n}|\langle c^nH^i | d^nH^j\rangle|^2\\
    &=\sum_{j,d^n}\langle c^nH^i|d^nH^j\rangle\langle d^n H^j|c^n H^i\rangle\\
    &=|K|^n\sum_j\langle c^nH^i|\rho_j^{\otimes n}|c^nH^i\rangle\\
    &=|K|^n\sum_j\prod_{l=1}^{n}\langle c_lH^i|\rho_j|c_lH^i\rangle.
\end{align}
Now we compute
\begin{align}\langle cH^i|\rho_j|cH^i\rangle=\frac{1}{|K|}\sum_{d\in K_j}|\langle cH^i|dH^j\rangle|^2=\frac{1}{|G||H|}\sum_d|cH^i\cap dH^j|^2=\frac{1}{|G||H|}\frac{|H|}{|H^i\cap H^j|}|H^i\cap H^j|^2,\end{align}
where the last equality is due to the standard fact from group theory that for any two subgroups $H_1, H_2\leq G$ and any $g_1,g_2\in G$, the intersection $g_1H_1\cap g_2H_2$ is either empty or a coset of $H_1\cap H_2$.

Therefore, 
\begin{align}
	D_{i,c^n}=\frac{|K|^n}{|G|^n}\sum_j|H^i\cap H^j|^n=\frac{|K|^n}{|G|^n}(|H|^n+\delta_i)=1+\frac{\delta_i}{|H|^n},
\end{align} 
where 
\begin{align}
	\delta_i=\sum_{j\neq i}|H^i\cap H^j|^n\leq m\max_{1\leq j<k\leq m}|H^j\cap H^k|^n=m\frac{|H|^n}{t^n}.
\end{align}
Hence
\begin{align}\ppgm\geq\frac{1}{m|K|^n}\sum_{1\leq i\leq m,c^n\in K_i^n}\frac{1}{D_{i,c^n}}\geq\frac{1}{m|K|^n}\cdot m|K|^n\frac{1}{1+\frac{m}{t^n}}=\frac{1}{1+\frac{m}{t^n}}
\end{align}
and the error probability is bounded above by $2\frac{m}{t^n}$, i.e. the sample complexity is at most $O(\frac{\log m}{\log t})$.
\end{proof}

\section{Conclusion}
\label{sec:conclusion}

In this paper we studied quantum state discrimination of geometrically uniform ensembles, a task appearing in various settings such as port-based teleportation \cite{ishizaka2008asymptotic,ishizaka2009quantum}, the hidden subgroup problem \cite{bacon2005semidirect,bacon2006dihedral,childs2010algebraic} or learning tasks \cite{arunachalam2020quantumcoupon,hadiashar2024optimal}.
In the first part of the paper we adopted a representation-theoretic viewpoint and discussed GU ensembles arising from an irreducible representation of a finite group.
This setting gives instances where the $\alpha$-power-pretty good measurement (PGM) \cite{tyson2009weighted,tyson2009two-sided} becomes increasingly better with increasing weight $\alpha$, and converges to the optimal measurement in the limit $\alpha\to \infty$.
Since the usual PGM corresponds to setting $\alpha=1$, this also provides an example where the PGM is provably sub-optimal, in contrast to tasks such as port-based teleportation \cite{studzinski2017port,leditzky2020} or various variants of the hidden subgroup problem \cite{bacon2005semidirect,bacon2006dihedral,moore2007conjugate} where the PGM is optimal.
In the second part we focused on pure-state GU ensembles, giving a new simplified proof of the result in \cite{eldar2004optimal} that the PGM is optimal in this setting, and a new derivation of the success probability of the PGM for solving the dihedral hidden subgroup problem \cite{bacon2006dihedral}.
We also made connections of this setting to the quantum coupon collector problem \cite{arunachalam2020quantumcoupon,hadiashar2024optimal}.
Finally, we considered arbitrary (mixed-state) GU ensembles in the many-copy setting and gave a compact lower bound on the success probability of the PGM.
This allowed us to rederive an upper bound on the sample complexity of the hidden subgroup problem that first appeared in \cite{hayashi2006hiddensubgroup}.

In the course of this paper we identified various learning tasks that feature GU state discrimination problems, most notably various versions of the hidden subgroup problem and the quantum coupon collector problem.
In future work we aim to explore further applications of our results to these and other learning tasks.
A major goal is to leverage the GU structure of the corresponding state discrimination task to derive strong bounds on the sample complexity of solving the task.
Another interesting question is to find quantum information-theoretic applications of the GU ensembles based on Werner state generators and permutation-invariant generators discussed in \Cref{sec:generators-symmetry}.

\paragraph*{Acknowledgments}
We acknowledge helpful discussions with Ashwin Nayak about \cite{hadiashar2024optimal}, with Shiliang Gao about \cref{eq: PBT decomposition}, and with the participants of the conference \emph{Beyond IID in Information Theory 12}, July 29 to August 2, 2024, University of Illinois Urbana-Champaign, during which part of this work was completed. This research was supported by a grant through the IBM-Illinois Discovery Accelerator Institute.
	
\printbibliography[heading=bibintoc]
\end{document}